\documentclass[11pt,a4paper]{article}
\usepackage{preprint}
\usepackage{mathtools,booktabs,array}
\usepackage{xcolor}
\hypersetup{colorlinks=true,linkcolor=blue,citecolor=blue,urlcolor=blue}

\newcommand{\F}{\mathbb F}
\newcommand{\E}{\mathbb E}
\newcommand{\Prb}{\mathbb P}
\newcommand{\1}{\mathbf 1}
\newcommand{\polylog}{\operatorname{polylog}}
\newcommand{\wtO}{\widetilde O}
\newcommand{\wtTheta}{\widetilde\Theta}

\title{Quantum Query Complexity of Persistence Statistics in Graph Zigzags}
\author{Cheng Xin\\[0.4em]\small Department of Computer Science\\\small California State University, Fresno\\\small\texttt{cxin@mail.fresnostate.edu}}
\date{}
\AtBeginDocument{\hypersetup{
  pdftitle={Quantum Query Complexity of Persistence Statistics in Graph Zigzags},
  pdfauthor={Cheng Xin},
  pdfsubject={Preprint: quantum query complexity of persistence statistics},
  pdfkeywords={zigzag persistence, persistence statistics, temporal graphs,
  quantum query complexity, randomized query complexity}}}

\begin{document}
\maketitle
\begin{abstract}
Temporal graph data is often stored as snapshots, and many uses of zigzag
persistence need only a scalar summary of bar lifetimes rather than the
barcode.  We study the query complexity of estimating such summaries from
snapshot-adjacency bits.

For graphs $G_1,\ldots,G_m$ on $n$ labeled vertices, let $\ell_b$ be the
snapshot lifetime of a degree-one bar $b$ of the intersection zigzag.  For a
probability generating function $\phi(x)=\E[x^R]$, the statistic
$F_\phi=\sum_b\phi(\ell_b/m)$ includes normalized degree-$r$ total persistence
and the mean generalized rank over a uniform time window.  Our starting point
is an exact identity: sample $R$, then $R$ uniform time indices; the expected
generalized rank of the window between their minimum and maximum equals
$F_\phi$.  For graphs that rank is the circuit rank of an intersection graph,
so a nonlinear barcode functional becomes an average of edge and component
counts, and no barcode is computed.

Without spectral-gap, homology-state, or QRAM assumptions, this gives a
quantum estimator with additive error $\varepsilon n$ and
$\widetilde O(\sqrt{m(K+n)}/\varepsilon)$ queries when a bound $K\ge F_\phi$ is
supplied, against $\widetilde O(m\min\{n^2,(K+n)/\varepsilon^2\})$
classically, and an adaptive quantum variant with the same instance
dependence.  These estimators are optimal in two regimes.  For every fixed
power weight $x^r$, $r\ge2$, and for the uniform-window mean, the worst-case
complexities are $\widetilde\Theta(n\sqrt m/\varepsilon)$ quantum and
$\Theta(n^2m)$ classical.  On sparse instances, under an explicit
split-leakage promise met by power and binomial weights of logarithmic degree
and the promise $F_\phi\le K$, they are $\widetilde\Theta(\sqrt{mK}/\varepsilon)$
and $\widetilde\Theta(m\min\{n^2,K/\varepsilon^2\})$.  The classical lower
bounds hold against fully adaptive algorithms, and fewer than $m$ such
statistics cannot determine the positive-lifetime histogram.  All bounds
concern snapshot access; with an explicit update stream, near-linear
full-barcode algorithms are known.
\end{abstract}

\section{Introduction}\label{sec:intro}

Zigzag persistence records how the topological features of a time-varying
space appear, disappear, and reappear~\cite{CarlssonDeSilva2010,CarlssonDeSilvaMorozov2009}.
Its output is a barcode, a multiset of intervals with one interval per
feature, whose length is the feature's lifetime.  In many applications the
barcode is only an intermediate object: what enters the downstream analysis is
a scalar summary of the lifetimes, such as total persistence
~\cite{CohenSteinerEtAl2010,Bubenik2015} or a linear representation whose
weight depends on persistence alone
~\cite{Bubenik2015,ChazalSilhouette,ChungLawson,DivolPolonik2019,WuKimRinaldo2024}.
This paper asks how many bits of a temporal graph a quantum or a randomized
algorithm must read to estimate such a summary, and how much of the barcode it
must implicitly compute.

We study graph zigzags, where the complete barcode is computable in
near-linear time from an explicit update stream~\cite{DeyHou2021,DeyHou2022},
but under a more primitive representation.  The temporal graph is a sequence
$G_1,\ldots,G_m$ of adjacency snapshots on $n$ labeled vertices, and a query
returns one snapshot bit: whether a given vertex pair is an edge at a given
time.  A quantum algorithm may query in superposition; a randomized algorithm
may query adaptively.  Nothing else is supplied: no update stream, no index of
edge lifetimes, no component oracle, and no quantum random-access memory.
This is the natural access model for snapshot data, and the one in which
sublinear estimation is a meaningful question, because even deciding whether
one edge is present throughout a window of $h$ snapshots is a search over $h$
bits.

\smallskip\noindent\textbf{The statistics.}
The intersection zigzag
$G_1\leftarrow G_1\cap G_2\rightarrow G_2\leftarrow\cdots\rightarrow G_m$ has a
degree-one barcode whose bars $b$ have snapshot lifetimes $\ell_b$.  For a
probability generating function $\phi(x)=\E[x^R]$ with
$R\in\{1,\ldots,r_{\max}\}$ we study $F_\phi(\mathbf G)=\sum_b\phi(\ell_b/m)$.
The family contains normalized degree-$r$ total persistence
$m^{-r}\sum_b\ell_b^r$, obtained from $R\equiv r$, and the mean generalized
rank over a uniformly random time window, obtained from a two-point law on
$\{1,2\}$.  The task is to output a number within additive error
$\varepsilon n$ of $F_\phi(\mathbf G)$, for $1/n\le\varepsilon\le1/4$, with
bounded failure probability.

\smallskip\noindent\textbf{The key idea.}
Sample $R$, then $R$ independent uniform times, and let $I$ be the window
between the smallest and the largest of them.  A bar of lifetime $\ell$
contains $I$ exactly when all $R$ samples fall inside the bar, which has
probability $(\ell/m)^R$; averaging over $R$ gives $\phi(\ell/m)$.  Summing
over bars,
\begin{equation}\label{eq:intro-identity}
F_\phi(\mathbf G)=\E_I\,\rho_{\mathbf G}(I),
\end{equation}
where $\rho_{\mathbf G}(I)$, the generalized rank of the window, is the number
of bars containing $I$ (Theorem~\ref{thm:extrema}); the window law does not
depend on the input.  For graphs the rank is the circuit rank of the
intersection graph $G_I=\bigcap_{t\in I}G_t$ (Lemma~\ref{lem:window-rank}), so
$F_\phi(\mathbf G)=\E_I|E(G_I)|-n+\E_I\,c(G_I)$, an average edge count plus an
average component count.  Each is a sublinear estimation problem, and testing
whether an edge survives $I$ is a search for a missing snapshot, which costs
$O(\sqrt{|I|})$ quantum queries and $|I|$ classical ones.  No barcode is ever
computed.

\smallskip\noindent\textbf{Results.}
Let $B_n=\binom{n-1}{2}$, the largest circuit rank of an $n$-vertex graph, so
$0\le F_\phi\le B_n$; $\wtO$ and $\wtTheta$ hide factors polylogarithmic in
$n$, $m$, $1/\varepsilon$, and the failure probability.

\begin{theorem}[Informal; Theorems~\ref{thm:q-supplied},
\ref{thm:q-adaptive}, and~\ref{thm:c-supplied}]
\label{thm:intro-upper}
Let $K\ge F_\phi(\mathbf G)$ be supplied.  An additive-$\varepsilon n$
estimate of $F_\phi(\mathbf G)$ can be computed with
$\wtO(\sqrt{m(K+n)}/\varepsilon)$ quantum and
$\wtO(m\min\{n^2,(K+n)/\varepsilon^2\})$ randomized snapshot queries.  Without
a supplied bound, a quantum algorithm attains the same expression with $K$
replaced by $F_\phi(\mathbf G)$, in expectation and with high probability,
under a deterministic cap of $\wtO(n\sqrt m/\varepsilon)$.
\end{theorem}

The quantum algorithm needs no spectral gap, no prepared homology state, and
no QRAM, assumptions that are central to quantum algorithms for Betti
numbers~\cite{LloydGarneroneZanardi2016,GunnKornerup2019,GyurikCadeDunjko2022,Hayakawa2022};
they are absent here because the target is a count.

\begin{theorem}[Informal; Theorem~\ref{thm:dense}]
\label{thm:intro-dense}
For every fixed power weight $\phi(x)=x^r$ with $r\ge2$, and for the
uniform-window mean, additive-$\varepsilon n$ estimation has quantum query
complexity $\wtTheta(n\sqrt m/\varepsilon)$ and randomized query complexity
$\Theta(n^2m)$.
\end{theorem}

No randomized algorithm improves on reading a constant fraction of the input,
even with error $n/4$; for constant $\varepsilon$ the quantum bound is the
square root of the randomized one, a Grover-type advantage that is provably no
larger.  The threshold $r\ge2$ is genuine: for $r=1$ the window is one snapshot
and the $\sqrt m$ factor disappears (Section~\ref{sec:dense}).

Theorem~\ref{thm:intro-upper} scales with $K$ rather than with $n^2$.  Whether
sparse instances are easier for every algorithm depends on one scalar of the
weight, its \emph{central split leakage}
\[
\lambda_{\phi,m}
=\max_{m/3\le t\le2m/3}
\left[\phi\!\left(\frac{t-1}{m}\right)+\phi\!\left(\frac{m-t}{m}\right)\right],
\]
the largest weight a full-lifetime bar retains after one central snapshot is
deleted and the bar breaks in two: the weight surviving on a decoy that an
algorithm must tell apart from a full bar.

\begin{theorem}[Informal; Theorem~\ref{thm:low-density}]
\label{thm:intro-sparse}
Let $256n\le K\le B_n/8$ and $1/n\le\varepsilon\le1/8$, and let $\phi$ satisfy
the leakage promise $\lambda_{\phi,m}B_n\le\varepsilon n/64$.  Under the promise
$F_\phi(\mathbf G)\le K$, additive-$\varepsilon n$ estimation has quantum
query complexity $\wtTheta(\sqrt{mK}/\varepsilon)$ and randomized query
complexity $\wtTheta(m\min\{n^2,K/\varepsilon^2\})$.
\end{theorem}

The promise is a property of the weight, enters only the lower bound, and is
met at logarithmic degree by the power weight $x^r$ and the binomial weight
$x((1+x)/2)^s$ (Corollaries~\ref{cor:power} and~\ref{cor:binomial}).  Small
leakage lets an adversary hide
$\Theta(n^2)$ decoy split bars inside the promise, so an algorithm must still
locate the $\Theta(\varepsilon n)$ full bars that separate two answers among
$\Theta(K)$ candidates, each verifiable only by a $\sqrt m$-cost search.  For
fixed $r$ the leakage is a constant and the decoys break the promise; whether
fixed-power statistics are easier on sparse instances is our main open
question.  Finally, fewer than $m$ statistics of this kind with rational
weights cannot determine even the histogram of positive bar lifetimes
(Theorem~\ref{thm:statistic-nonrecovery}), so the problem is not a disguised
barcode computation.

\smallskip\noindent\textbf{Contributions.}
(i) A sampling representation of nonlinear persistence statistics,
identity~\eqref{eq:intro-identity}, valid for every finite zigzag module.
(ii) Quantum and classical estimators whose cost scales with the statistic's
value, without spectral, state-preparation, or memory assumptions.
(iii) Matching lower bounds on the statistic itself, in both models, in a
dense and in a leakage-controlled sparse regime; the randomized bounds hold
against fully adaptive querying through an analyst-only censoring device.
(iv) Split leakage as the governing parameter: a choice made for statistical
reasons, the suppression of short bars, has a query-complexity consequence.
The summaries and the quantum primitives are established; our contribution is
the representation, the algorithms, and the classification.

\smallskip\noindent\textbf{Scope and organization.}
The $\sqrt m$ factor and the lower bounds are statements about snapshot
access: with an explicit update stream or a searchable edge-lifetime index the
complete barcode is computable in near-linear time~\cite{DeyHou2021}, and a
stored bit tensor is not a coherent unit-cost oracle.  Our results are query
classifications; gate work, depth, workspace, and output are reported
separately.  Section~\ref{sec:setup} fixes the model,
Sections~\ref{sec:identity} to~\ref{sec:histogram} state the results with
proof sketches, and Sections~\ref{sec:related} and~\ref{sec:discussion} give
related work and open problems.  Complete proofs are in
Appendices~\ref{app:identity} to~\ref{app:weights}, grouped by section;
Appendix~\ref{app:checks} summarizes the finite sanity checks used during
proof verification.

\section{Problem Setup}\label{sec:setup}

\noindent\textbf{Temporal graph oracle.}
Let $\mathbf G=(G_1,\ldots,G_m)$ be simple undirected graphs on a common
labeled vertex set $[n]$.  The input is the bit oracle
\begin{equation}\label{eq:oracle}
 O_{\mathbf G}\lvert t,u,v,z\rangle
 =\lvert t,u,v,z\oplus \1_{\{u,v\}\in E(G_t)}\rangle ,
\end{equation}
queried coherently by a quantum algorithm and bit by bit by a classical one.
No adjacency list, update list, edge-lifetime index, component oracle, or
QRAM is supplied, and graph-dependent preprocessing is charged for its
queries.  Query complexity counts calls to~\eqref{eq:oracle}; gate work,
depth, workspace, and output are reported separately, with $C_O$ and $D_O$
the gate cost and depth of one oracle call.

\smallskip\noindent\textbf{Intersection zigzag and window rank.}
Regard every $G_t$ as a one-dimensional simplicial complex over an arbitrary
fixed field $\F$ and define
$\mathcal Z(\mathbf G):G_1\leftarrow G_1\cap G_2\rightarrow G_2\leftarrow\cdots\rightarrow G_m$.
For $1\le a\le b\le m$, the window $[a,b]$ contains the snapshot nodes
$G_a,\ldots,G_b$ and all intervening intersection nodes; put
$G_{[a,b]}=\bigcap_{t=a}^bG_t$ and
$\rho_{\mathbf G}([a,b])=\operatorname{rank}(\lim H_1\to\operatorname{colim}H_1)$
for the restricted zigzag.

\begin{lemma}[Window-rank identity]\label{lem:window-rank}
For every field $\F$ and every window $[a,b]$,
\[
\rho_{\mathbf G}([a,b])
=\beta_1(G_{[a,b]};\F)
=|E(G_{[a,b]})|-n+c(G_{[a,b]}),
\]
where isolated vertices are included in $c$.
\end{lemma}

Cycle spaces of graphs are subspaces of the edge-chain space of the union
graph, so a compatible limit vector is one cycle supported in every snapshot
of the window and the canonical map is injective on it
(Appendix~\ref{app:identity}).

\smallskip\noindent\textbf{Bar lifetimes and persistence statistics.}
A finite type-$A$ zigzag representation over a field has an interval
decomposition.  For an interval summand $b$ of $H_1(\mathcal Z(\mathbf G))$,
let $\ell_b\in\{0,1,\ldots,m\}$ be the number of \emph{snapshot nodes} in its
support; intersection nodes do not count, and bars supported only on junction
nodes have $\ell_b=0$.

\begin{definition}[PGF-weighted persistence statistic]
Let $R$ be an input-independent random variable supported on
$\{1,\ldots,r_{\max}\}$, with probabilities $a_r\ge0$ summing to one.  Its
probability generating function is
$\phi(x)=\E[x^R]=\sum_{r=1}^{r_{\max}}a_rx^r$, $0\le x\le1$.  The associated
graph-zigzag persistence statistic is $F_\phi(\mathbf G)=\sum_b\phi(\ell_b/m)$.
\end{definition}

In the power case $\phi_r(x)=x^r$, $F_{\phi_r}=m^{-r}\sum_b\ell_b^r$ is
normalized degree-$r$ total persistence~\cite{CohenSteinerEtAl2010,Bubenik2015},
a linear representation whose weight depends only on
persistence~\cite{DivolPolonik2019,WuKimRinaldo2024}.  The algorithm receives
reversible classical and coherent samplers for $R$ with declared gate cost
$C_\phi$, depth $D_\phi$, and workspace $W_\phi$.  Since $F_\phi$ is an
expectation of a cycle rank, $0\le F_\phi\le B_n=\binom{n-1}{2}$.

\begin{definition}[Estimation problem]
Given $n,m,\phi,\varepsilon,\delta$ and oracle access to $\mathbf G$, output one
classical number $\widehat F$ such that
$\Prb\bigl[|\widehat F-F_\phi(\mathbf G)|\le\varepsilon n\bigr]\ge1-\delta$.
The main regime is $1/n\le\varepsilon\le1/4$.
\end{definition}

\section{Persistence Statistics from Random Windows}\label{sec:identity}

\begin{theorem}[Random-extrema identity]\label{thm:extrema}
Let $M$ be any finite type-$A$ zigzag module over a field, with $m$ distinguished
snapshot nodes and interval-summand snapshot lifetimes $\ell_b$.  Sample
$R\sim(a_r)$; conditional on $R=r$, sample $T_1,\ldots,T_r$ independently and
uniformly from $[m]$.  Let
\[
I_R=[\min_iT_i,\max_iT_i]
\]
with all zigzag nodes between these snapshots.  Then
\[
\E\,\rho_M(I_R)=\sum_b\phi(\ell_b/m).
\]
In particular, for $M=H_1(\mathcal Z(\mathbf G))$,
\[
F_\phi(\mathbf G)=\E\,\beta_1(G_{I_R}).
\]
\end{theorem}

\begin{proof}[Proof sketch]
An interval summand contributes one to the generalized rank of $I_R$ exactly
when the window lies in its support, that is, when all $r$ sampled indices
fall among the bar's $\ell_b$ contiguous snapshot nodes, an event of
probability $(\ell_b/m)^r$; generalized rank is additive under direct sums,
so average over $R$ and apply Lemma~\ref{lem:window-rank}.
\end{proof}

Appendix~\ref{app:identity} records two consequences: a nonnegative
polynomial weight $\psi$ with coefficient sum $C$ is $C$ times a PGF weight,
so every bound below applies with $\varepsilon/C$ in place of $\varepsilon$
(Corollary~\ref{cor:scale}), and the uniform-window mean over the $m(m+1)/2$
nonempty windows is the PGF weight $\phi^{\mathrm{unif}}_m(x)=(x+mx^2)/(m+1)$,
that is, $R=1$ with probability $1/(m+1)$ and $R=2$ otherwise
(Corollary~\ref{cor:uniform}).

\section{Quantum and Classical Estimators}\label{sec:estimators}

We use quantum approximate counting~\cite{BHMT,AaronsonRall2020} and search
with an unknown number of solutions~\cite{BBHT}, which implements the
intersection-edge bit of a window of length at most $h$ with
$\widetilde O(\sqrt h)$ snapshot queries; nested calls are amplified and
uncomputed so that a hybrid argument applies (Appendix~\ref{app:estimators}).

\begin{theorem}[Supplied-density quantum estimator]\label{thm:q-supplied}
Let $\phi$ be supplied by efficient classical and coherent sampling circuits,
and suppose a number $K\ge F_\phi(\mathbf G)$ is supplied.  For
$1/n\le\varepsilon\le1/4$ and $0<\delta<1/3$, there is a quantum algorithm that
uses
\[
\widetilde O\!\left(\frac{\sqrt{m(K+n)}}{\varepsilon}\right)
\]
snapshot queries and returns an additive-$\varepsilon n$ estimate with
probability at least $1-\delta$.  This is a deterministic worst-case query
budget once $K$ is supplied.
\end{theorem}

\begin{proof}[Proof sketch]
With $\bar e=\E_I|E(G_I)|$ and $\bar c=\E_Ic(G_I)$, $F=\bar e-n+\bar c$ and
$F\le\bar e\le K+n-1$, so amplitude estimation of the marked probability of a
uniform potential edge paired with a sampled window reaches error
$\varepsilon n/4$ with $\widetilde O(\sqrt{K+n}/\varepsilon)$ evaluations of
the survival predicate.  For components, $c(H)/n=\E_v[1/|C_H(v)|]$; at dyadic
scale $2^j$ a reversible exploration halting after $2^{j+1}$ discoveries
costs $\widetilde O(2^j\sqrt n)$ adjacency queries while the precision needed
there is proportional to $\varepsilon2^j$, so each of the
$O(\log(1/\varepsilon))$ scales costs $\widetilde O(\sqrt n/\varepsilon)$
(Proposition~\ref{prop:q-dyadic-full}); every adjacency query of $G_I$ is a
search for a missing snapshot, charged $\widetilde O(\sqrt m)$ snapshot queries.
\end{proof}

\begin{theorem}[Adaptive quantum estimator]\label{thm:q-adaptive}
No upper bound on $F=F_\phi(\mathbf G)$ need be supplied.  There is a quantum
algorithm with success probability at least $1-\delta$, deterministic cap
\[
\widetilde O(n\sqrt m/\varepsilon),
\]
expected snapshot-query complexity
\[
\widetilde O\!\left(\frac{\sqrt{m(F+n)}}{\varepsilon}\right),
\]
and the same instance-sensitive bound outside an event of probability
$O(\delta/n^3)$.
\end{theorem}

\begin{proof}[Proof sketch]
Append $n$ known marked dummies, so the marked mass $E'=\bar e+n$ satisfies
$F+n\le E'\le F+2n$; $O(\sqrt n)$ amplitude-estimation iterations then give a
constant-factor estimate of $E'$ with failure $\Theta(\delta/n^3)$, which
fixes the budget of the final estimator, capped by the dense budget.
\end{proof}

\begin{theorem}[Supplied-density randomized estimator]\label{thm:c-supplied}
Under the assumptions of Theorem~\ref{thm:q-supplied}, there is a randomized
classical algorithm using
\[
\widetilde O\!\left(
 m\min\left\{n^2,\frac{K+n}{\varepsilon^2}\right\}
\right)
\]
snapshot queries and returning an additive-$\varepsilon n$ estimate with
probability at least $1-\delta$.
\end{theorem}

\begin{proof}[Proof sketch]
Bernstein sampling with variance bound $\bar e\le K+n$ needs
$\widetilde O((K+n)/\varepsilon^2)$ samples of at most $m$ bits each; the
component term uses the dyadic decomposition with a breadth-first search
stopped after $2^{j+1}$ discoveries and Hoeffding sampling
(Proposition~\ref{prop:c-dyadic}); or read all $O(n^2m)$ bits.  A
dummy-density bootstrap gives the same bound with $F$ in place of $K$ in
expectation (Proposition~\ref{prop:c-adaptive}).
\end{proof}

\section{Matching Bounds in the Dense Regime}\label{sec:dense}

For $r\ge1$ put $S_r(\mathbf G)=\sum_b(\ell_b/m)^r$, and let
$A(\mathbf G)=\sum_b\ell_b(\ell_b+1)/(m(m+1))$ be the uniform-window mean.

\begin{theorem}[Dense lifetime-moment classification]\label{thm:dense}
Let $n\ge192$, $m\ge3$, and $1/n\le\varepsilon\le1/4$.  For every fixed
integer $r\ge2$, estimating either $S_r$ or $A$ to additive error
$\varepsilon n$ has bounded-error query complexities
\[
Q_2=\widetilde\Theta(n\sqrt m/\varepsilon),
\qquad
R_2=\Theta(n^2m)
\]
in the coherent and randomized raw-snapshot models, respectively.
\end{theorem}

\begin{proof}[Proof sketch]
Upper bounds: Theorems~\ref{thm:q-supplied} and~\ref{thm:c-supplied} with
$K=B_n$.  Both lower bounds use one gadget: a spanning star in every snapshot
and one optional chord per leaf pair, giving $B_n$ independent temporal rows;
an all-one row is a bar of lifetime $m$, and a row with one zero at a central
time $t\in\mathcal T_m=\{\lceil m/3\rceil,\ldots,\lfloor2m/3\rfloor\}$,
$L=|\mathcal T_m|\ge m/4$, is two bars retaining at most $5/9$ of a full bar's
weight for every $x^r$, $r\ge2$, and for $A$.  Quantum: erasing
$d=\lfloor6\varepsilon n\rfloor+1$ zeros from inputs with $\lfloor B_n/2\rfloor$
all-one rows changes the output by more than $2\varepsilon n$, and the adversary
lemma for real-valued outputs with pairwise-disjoint acceptance intervals
(Proposition~\ref{prop:metric-adversary}) gives $\Omega(n\sqrt m/\varepsilon)$
from the relation degrees.  Randomized: rows independently all-one with
probability $1/2\pm48/n$ separate the means by $\Theta(n)$ against variance
$O(n^2)$, while a fresh central query in a row contributes only $O(1/(n^2L))$
to the transcript divergence, so Pinsker's inequality forces
$\Omega(n^2L)=\Omega(n^2m)$ queries (Appendix~\ref{app:dense}).
\end{proof}

For $\phi(x)=x$ the window is a single snapshot, the estimators run with window
length one at cost $\widetilde O(\sqrt{K+n}/\varepsilon)$ with no $\sqrt m$
factor, and a central split removes only $1/m$ of a bar's weight, so the
construction yields no lower bound: the threshold $r\ge2$ is genuine.

\section{Matching Bounds under a Low-Density Leakage Promise}\label{sec:sparse}

\begin{definition}[Central split leakage]\label{def:leakage}
For a probability-generating weight function $\phi$ and $m\ge3$, define
$\lambda_{\phi,m}=\max_{t\in\mathcal T_m}\bigl[\phi\bigl(\tfrac{t-1}{m}\bigr)+\phi\bigl(\tfrac{m-t}{m}\bigr)\bigr]$.
It is the largest weighted mass retained when a full bar is split by one
missing central snapshot.
\end{definition}

The leakage promise below enters only the lower bounds; the upper bound of
Theorem~\ref{thm:q-supplied} holds for every weight.  The constants are
conservative.

\begin{theorem}[Leakage-controlled low-density classification]
\label{thm:low-density}
Assume
\[
n\ge2^{14},\qquad m\ge16,
\qquad \frac1n\le\varepsilon\le\frac18,
\]
and let $K$ be an integer satisfying
\[
256n\le K\le\frac18\binom{n-1}{2}.
\]
Let $\phi$ be a coherently and classically sampleable probability-generating
weight function satisfying
\begin{equation}\label{eq:leakage-promise}
\lambda_{\phi,m}\binom{n-1}{2}\le\frac{\varepsilon n}{64}.
\end{equation}
On the promise $F_\phi(\mathbf G)\le K$, additive-$\varepsilon n$
estimation has query complexities, up to polylogarithmic factors,
\[
Q_2=\widetilde\Theta\!\left(\frac{\sqrt{mK}}\varepsilon\right)
\]
and
\[
R_2=\widetilde\Theta\!\left(
 m\min\left\{n^2,\frac K{\varepsilon^2}\right\}
\right).
\]
The upper bounds are worst-case bounds because $K$ is supplied.  The adaptive
no-$K$ result remains the separate statement of
Theorem~\ref{thm:q-adaptive}; it is not part of this minimax equality.
\end{theorem}

\begin{proof}[Proof sketch]
Upper bounds: Theorems~\ref{thm:q-supplied} and~\ref{thm:c-supplied} with
$K\ge256n$.  Quantum lower bound: the gadget of Theorem~\ref{thm:dense} with
$\lfloor K/4\rfloor$ all-one rows and every other row split; by
\eqref{eq:leakage-promise} the split population contributes at most
$\varepsilon n/64$, so all related inputs satisfy $F_\phi\le K$, and the same
relation with $d=\lceil4\varepsilon n\rceil$ gives
$\Omega(\sqrt{mK}/\varepsilon)$.  Randomized lower bound: rows independently
all-one with probabilities $p_-=K/(4B_n)$ and $p_+=(K/4+D)/B_n$,
$D=16(\varepsilon n+\sqrt K)$;
Chebyshev's inequality places both distributions, up to probability $1/32$,
on promise inputs whose statistics lie more than $\varepsilon n$ from a common
threshold.  Because $p_-=\Theta(K/n^2)$ is small, late queries in a row are
ill-conditioned; Proposition~\ref{prop:censored-kl} censors each row after
$\lfloor L/2\rfloor$ one-answers by handing the analysis transcript, but not
the algorithm, the row's type, and bounds the divergence of any adaptive
$Q$-query transcript by $128Q\Delta^2/(p_-L)$, $\Delta=D/B_n$, so Pinsker's
inequality yields $\Omega(KB_nL/D^2)=\Omega(m\min\{n^2,K/\varepsilon^2\})$
(Appendix~\ref{app:sparse}).
\end{proof}

\section{Weight Functions and Resources}\label{sec:weights}

\begin{corollary}[Logarithmic power weights]\label{cor:power}
For $\phi_r(x)=x^r$, $\lambda_{\phi_r,m}\le2(2/3)^r$.  Consequently the leakage
promise~\eqref{eq:leakage-promise} holds whenever
$r\ge\lceil\log_{3/2}(128B_n/(\varepsilon n))\rceil$.  The simpler sufficient
choice $r\ge\lceil\log_{3/2}(64n/\varepsilon)\rceil$ follows from $B_n\le n^2/2$.
\end{corollary}

Let $R=1+\operatorname{Bin}(s,1/2)$, with probability generating function
$\phi_s^{\mathrm{bin}}(x)=x\bigl(\tfrac{1+x}{2}\bigr)^s$; it smoothly emphasizes
bars close to full lifetime and uses at most $s+1$ time samples.

\begin{corollary}[Binomial soft-long-bar weight]\label{cor:binomial}
For $\phi_s^{\mathrm{bin}}$, $\lambda_{\phi_s^{\mathrm{bin}},m}\le\frac43(5/6)^s$.
Thus~\eqref{eq:leakage-promise} holds whenever
\[
s\ge\left\lceil\log_{6/5}\frac{256B_n}{3\varepsilon n}\right\rceil,
\qquad\text{and it is sufficient to take}\qquad
s\ge\left\lceil\log_{6/5}\frac{128n}{3\varepsilon}\right\rceil.
\]
\end{corollary}

Both follow because a central split leaves two pieces of normalized length at
most $2/3$.  Window preparation for either weight uses
$C_\phi=\widetilde O(r_{\max}\log m)$ gates and $O(r_{\max}\log m)$ qubits with
no graph data.  The supplied-$K$ quantum algorithm has gate work
$\widetilde O(\sqrt{m(K+n)}(C_O+\log m)/\varepsilon)$ for the oracle calls plus
$\widetilde O(\sqrt{K+n}\,C_\phi/\varepsilon+\sqrt n/\varepsilon^2)$ for window
preparation and list membership, polylogarithmic live workspace beyond the
$O(\varepsilon^{-1}\log n)$ qubits of the explicit discovered list, depth at
most its gate work, and an $O(\log(n/\varepsilon))$-bit output
(Appendix~\ref{app:weights}).

\section{Limits of Recovering the Positive-Lifetime Histogram}\label{sec:histogram}

All statistics considered here discard the positions of bars and retain only
their lifetimes. Write $(h_1,\ldots,h_m)$ for the histogram of bars with a
positive snapshot lifetime. Bars with zero snapshot lifetime are invisible to
every PGF weight used here. Even this positive-lifetime histogram requires many
independent weighted statistics.

\begin{theorem}[Fewer than $m$ rational PGF weights do not determine the
positive-lifetime histogram]\label{thm:statistic-nonrecovery}
Fix $m\ge2$ and rational-coefficient probability-generating weight functions
$\phi_1,\ldots,\phi_k$ with $k<m$.  There exist two graph snapshot sequences,
possibly on a larger common vertex set, whose $H_1$ zigzag barcodes have
different positive-lifetime histograms but satisfy
\[
F_{\phi_j}(\mathbf G)=F_{\phi_j}(\mathbf G')
\qquad(1\le j\le k).
\]
The sequences may be chosen with a common spanning star in every snapshot.
\end{theorem}

\begin{proof}[Proof sketch]
The matrix $A_{j\ell}=\phi_j(\ell/m)$ has a nonzero integer kernel vector
$z$ with entries of both signs, since all entries of $A$ are positive;
$z=z^+-z^-$ gives two distinct positive-lifetime histograms with $Az^+=Az^-$,
each realized by a spanning star with one chord per requested bar present on a
contiguous block of the requested length (Appendix~\ref{app:weights}).
\end{proof}

\section{Related Work}\label{sec:related}

\paragraph{Persistence, zigzags, and scalar summaries.}
Persistent homology was founded in~\cite{ELZ2002,ZomorodianCarlsson2005}
(modern treatment: \cite{DeyWang2022}); zigzag persistence and its connection
to time-varying data in~\cite{CarlssonDeSilva2010,CarlssonDeSilvaMorozov2009},
with efficient matrix algorithms in~\cite{MilosavljevicMorozovSkraba2011}.
Landscapes, silhouettes, persistence curves, and linear representations give
broad families of barcode summaries
~\cite{Bubenik2015,ChazalSilhouette,ChungLawson,DivolPolonik2019,WuKimRinaldo2024},
and total persistence is a standard lifetime aggregate
~\cite{CohenSteinerEtAl2010,Bubenik2015}; our statistics are scalar linear
representations whose weights depend only on persistence, not new invariants.

\paragraph{Algorithms for zigzag and multiparameter persistence.}
The graph zigzag barcode is computable in near-linear time from an
explicit update sequence~\cite{DeyHou2021}, with faster general algorithms
in~\cite{DeyHou2022}; related work treats barcode updates, absolute and
relative zigzags, optimal persistent cycles, graph-persistence updates, and
representatives~\cite{DeyHouUpdate,DeyHouAssociation,DeyHouOptimal,DeyHouParsa2023,
DeyHouMorozov2025,DeyHouMorozov2026}.  These are the comparators when an update
stream is supplied; our snapshot-bit lower bounds do not apply there.
Multiparameter modules admit bottleneck-distance computation, rectangular
approximation, and decomposition~\cite{DeyXin2018,DeyXin2019,DeyXin2021,DeyXin2022};
generalized-rank information serves vectorization and
learning~\cite{XinEtAlGRIL2023,XinEtAlTopInG2025,MukherjeeEtAlDGRIL2026}; a
recent preprint proposes an unfolding from generalized ranks to zigzag
persistence~\cite{DeyXinUnfolding}, on which no theorem here depends; and the
two-parameter generalized-rank invariant is computed via zigzags, with its
diagrams, sparsification, and size studied
in~\cite{DeyKimMemoli2024,KimMemoli2021,CarriereKimKim2025,KimKimLee2025}.
Output-sensitive persistence~\cite{ChenKerber2013}, clearing and chunk
compression~\cite{BauerKerberReininghaus2014}, streaming towers
~\cite{KerberSchreiber2019}, nested dissection~\cite{KerberSheehySkraba2016},
multiparameter decomposition~\cite{DeyJendrysiakKerber2025}, and
matrix-multiplication-time persistence~\cite{MorozovSkraba2025} give a stronger
classical context than generic reduction; our comparator is the strongest
algorithm under the same raw oracle and output contract.

\paragraph{Quantum algorithms for topological data analysis and query primitives.}
Quantum TDA began with Betti-number estimation~\cite{LloydGarneroneZanardi2016},
whose hidden costs were identified in~\cite{GunnKornerup2019}; later work
studies persistent Betti numbers, quantum persistent homology, and routes to
provable advantage~\cite{Hayakawa2022,AmeneyroMaroulasSiopsis2024,GyurikCadeDunjko2022},
dequantization and complexity-theoretic analyses delimit when such advantages
are possible~\cite{ApersGriblingSenSzabo2023,SchmidhuberLloyd2023,BerryEtAl2024},
and recent algorithms reduce qubits or obtain speedups under explicit
topological and spectral promises~\cite{McArdleGilyenBerta2026,GyurikEtAl2026,HayakawaChenHsieh2026}.
Our target and access contract differ: one additive scalar from a temporal
graph, with no spectral-gap, homology-state, or QRAM assumption.
Static circuit rank can be estimated in the adjacency-matrix model~\cite{DKW};
our upper bounds use established search and counting
machinery~\cite{BBHT,BHMT,AaronsonRall2020,HamoudiMagniez2019}, related quantum
graph algorithms study connectivity under adjacency
access~\cite{DurrHeiligmanHoyerMhalla2006,JarretJefferyKimmelPiedrafita2018},
and the lower-bound ingredients are established~\cite{Ambainis2000,BunKothariThaler}.
Quantum frequency moments concern powers of item frequencies in a different
input model with a relative-error objective~\cite{MontanaroMoments}.  To our
knowledge, neither the random-extrema representation nor a query
classification of a persistence statistic under snapshot access appears in
prior work; the nearest results are~\cite{DKW}, the update-stream
algorithms~\cite{DeyHou2021,DeyHou2022}, and the persistence-curve
frameworks~\cite{ChungLawson,DivolPolonik2019}, which contain the invariant
family but do not address query complexity.

\section{Discussion and Open Problems}\label{sec:discussion}

The algorithms output one additive scalar, not a barcode, representatives, a
bottleneck approximation, or a guarantee retaining each rare long bar; an
allowance $\varepsilon n\ge1$ can hide one full-length bar, and conversely
Theorem~\ref{thm:statistic-nonrecovery} shows that a few scalars cannot
reconstruct the barcode.  The static identity and the quantum subroutines are
established; what is specific here is the random-extrema identity, the
split-bar reduction that changes the averaged output directly, the leakage
theorem tying instance-sensitive optimality to the weight's response to
broken bars, and the randomized bound against fully adaptive probing.  Two questions would substantially extend these results.
\begin{enumerate}
\item For $x^r$ with fixed $r\ge2$ and $F_\phi\le K\ll n^2$, is the complexity
$\widetilde\Theta(\sqrt{mK}/\varepsilon)$, or do constant-leakage weights admit
faster algorithms on sparse instances?
\item Lemma~\ref{lem:window-rank} uses the absence of two-dimensional faces; is
there a sampling representation of degree-$k$ statistics of simplicial
zigzags with a sublinear estimator?
\end{enumerate}

\appendix

\section{Proofs for Sections~\ref{sec:setup} and~\ref{sec:identity}}\label{app:identity}

\begin{proof}[Proof of Lemma~\ref{lem:window-rank}]
All graph cycle spaces are literal subspaces of the edge-chain space of the
union graph, because there are no degree-two boundaries.  A compatible vector
in the limit is therefore one cycle supported in every snapshot of the
window, equivalently a cycle of $G_{[a,b]}$.  Conversely, every such cycle is a
compatible family.  Its image in the colimit is nonzero whenever the cycle is
nonzero: the maps into the common ambient edge-chain space agree on the
zigzag and factor the canonical map.  Hence the canonical map is injective on
the common cycle space.  The last equality is Euler's formula for a graph.
\end{proof}

\begin{proof}[Proof of Theorem~\ref{thm:extrema}]
Fix an interval summand $b$ and condition on $R=r$.  Its contribution to the
generalized rank of $I_R$ is one exactly when the entire sampled window lies in
its support.  Since an interval support is contiguous, this is equivalent to
all $r$ sampled snapshot indices lying among the $\ell_b$ snapshot nodes of the
bar.  The probability is $(\ell_b/m)^r$.  If $\ell_b=0$, the probability is
zero because $r\ge1$.  Generalized rank is additive under direct sums, so
\[
\E[\rho_M(I_R)\mid R=r]=\sum_b(\ell_b/m)^r.
\]
Averaging over $R$ proves the first identity.  The graph statement follows
from Lemma~\ref{lem:window-rank}.
\end{proof}

\begin{corollary}[Nonnegative polynomials and normalization]\label{cor:scale}
Let $\psi(x)=\sum_{r=1}^{r_{\max}}c_rx^r$ with $c_r\ge0$ and
$C=\sum_rc_r>0$.  Put $\phi=\psi/C$.  Then
\[
\sum_b\psi(\ell_b/m)=C\,F_\phi(\mathbf G).
\]
Thus additive error $\varepsilon n$ for the unnormalized statistic requires
additive error $(\varepsilon/C)n$ for $F_\phi$: every bound below applies
with $\varepsilon$ replaced by $\varepsilon/C$, so $1/\varepsilon$ terms scale
by $C$, $1/\varepsilon^2$ terms scale by $C^2$, and the precision regime
becomes $C/n\le\varepsilon$.  A constant term $c_0$ is excluded: it
would count all bars, including bars that contain no sampled snapshot, and is
not represented by a nonempty random window.
\end{corollary}

\begin{corollary}[Uniform windows give a PGF-weighted statistic]\label{cor:uniform}
Let $I$ be uniform over the $m(m+1)/2$ nonempty snapshot windows.  Then
\[
\E_I\rho(I)=\sum_b\frac{\ell_b(\ell_b+1)}{m(m+1)}
=F_{\phi^{\mathrm{unif}}_m},
\]
where
\[
\phi^{\mathrm{unif}}_m(x)=\frac{x+mx^2}{m+1}.
\]
Equivalently, $R=1$ with probability $1/(m+1)$ and $R=2$ with probability
$m/(m+1)$.
\end{corollary}

\begin{proof}[Proof of Corollary~\ref{cor:uniform}]
A bar supported on $\ell$ consecutive snapshot nodes contains exactly
$\ell(\ell+1)/2$ nonempty snapshot windows.  The polynomial identity follows
from $x=\ell/m$.
\end{proof}

\section{Proofs for Section~\ref{sec:estimators}}\label{app:estimators}

\subsection{Dyadic reciprocal-size decomposition and proofs of the estimators}

For a graph $H$ and vertex $v$, let $s_H(v)=|C_H(v)|$.  Then
\[
\frac{c(H)}n=\E_v\frac1{s_H(v)}.
\]
This reciprocal-component identity and its use with bounded local exploration
are established in sublinear minimum-spanning-tree estimation~\cite{ChazelleRubinfeldTrevisan2005};
we adapt the implementation to adjacency-bit access and coherent mean estimation.
Put
\[
J=\left\lceil\log_2\frac{16}{\varepsilon}\right\rceil,
\qquad a_j=2^j,
\qquad
Z_j(v)=\frac{a_j}{s_H(v)}\1_{a_j\le s_H(v)<2a_j}.
\]
Then
\begin{equation}\label{eq:dyadic}
\frac{c(H)}n=\sum_{j=0}^{J-1}\frac{\E Z_j}{a_j}+R,
\qquad 0\le R\le2^{-J}\le\varepsilon/16.
\end{equation}
A reversible exploration that halts once $2a_j$ vertices are discovered
costs $\widetilde O(a_j\sqrt n)$ adjacency queries.  Estimating $\E Z_j$ to error
$\alpha_j=\varepsilon a_j/(32J)$ by amplitude estimation~\cite{BHMT} makes the factors
$a_j$ cancel.  Summing the levels gives
\begin{equation}\label{eq:qcomponents}
\widetilde O(\sqrt n/\varepsilon)
\end{equation}
adjacency queries for additive $O(\varepsilon n)$ component error.  The
sampler is a fixed reversible circuit whose output expectation differs from
the ideal one by at most the budgeted failure probability of its internal
searches, and the analysis estimates the mean of this actual circuit.

\begin{proof}[Proof of Theorem~\ref{thm:q-supplied}]
Sample the random-extrema window $I$ coherently.  Write
\[
\bar e=\E_I|E(G_I)|,
\qquad
\bar c=\E_Ic(G_I),
\qquad
F=F_\phi(\mathbf G).
\]
By Lemma~\ref{lem:window-rank} and Theorem~\ref{thm:extrema},
\[
F=\bar e-n+\bar c,
\qquad
F\le\bar e\le F+n-1\le K+n-1.
\]
Amplitude-estimate the marked probability of a uniform potential edge paired
with the sampled window.  The variance-sensitive counting bound gives
$\widetilde O(\sqrt{K+n}/\varepsilon)$ ideal intersection-edge queries for
additive error $\varepsilon n/4$.  Apply the dyadic component estimator to the
joint distribution of a window and a uniform vertex.  Equation
\eqref{eq:qcomponents} gives $\widetilde O(\sqrt n/\varepsilon)$ ideal
intersection-edge queries and component error below $\varepsilon n/2$.
Every intersection-edge query is an AND over at most $m$ snapshot bits and is
implemented in $\widetilde O(\sqrt m)$ snapshot queries.  The edge term absorbs
the component term because $K+n\ge n$.  Allocate failure and clean-oracle
approximation budgets across all calls and combine
$\widehat e-n+\widehat c$.
\end{proof}

\begin{proof}[Proof of Theorem~\ref{thm:q-adaptive}]
Let $M=\binom n2$ and append $n$ known marked dummy labels to the potential-edge
domain.  On a sampled window, mark a real edge when it survives the window.
The resulting average marked mass is
\[
E'=\bar e+n,
\qquad
F+n\le E'\le F+2n.
\]
It need not be an integer; amplitude estimation acts on its success
probability $E'/(M+n)$.  Since $E'\ge n$, $O(\sqrt n)$ amplitude-estimation
iterations give a constant-factor estimate.  Median amplification reduces the
total bootstrap failure, including implementation error, to
$\delta_0=\Theta(\delta/n^3)$.  From the estimate construct a clipped bound
$U$ such that on the good event $E'\le U\le(5/2)E'$ and always
$n\le U\le2(M+n)$.  Run the final edge estimator with its budget chosen from
$U$, add the dyadic component estimator, and cap by the dense budget.  On the
good event the cost has the stated dependence on $F$; on every branch it has
the dense cap.  Multiplying the cap by $\delta_0$ makes the bad-event
contribution negligible in expectation.  The bootstrap's \emph{total}
circuit error must be $O(\delta_0)$; constant simulation error would invalidate
the instance-sensitive expected bound.
\end{proof}

\begin{proposition}[Classical dyadic component estimation]\label{prop:c-dyadic}
Let $\Pi$ be an input-independent window distribution supported on windows of
length at most $h$.  The mean component count $\E_{I\sim\Pi}c(G_I)$ can be
estimated to additive error $\varepsilon n/2$, with failure at most $\delta$,
using
\[
\widetilde O(nh/\varepsilon^2)
\]
raw snapshot queries.
\end{proposition}

\begin{proof}[Proof of Proposition~\ref{prop:c-dyadic}]
Use~\eqref{eq:dyadic} jointly over a sampled window and a uniform vertex.  For
level $j$, perform an ordinary breadth-first search but stop after discovering
$2a_j$ vertices.  In the adjacency-matrix model this scans at most $2a_j$
rows, hence uses $O(na_j)$ intersection-edge queries; evaluating each such bit
uses at most $h$ snapshot queries.  The returned variable $Z_j$ lies in
$[0,1]$.  Hoeffding estimation to error
$\alpha_j=\varepsilon a_j/(32J)$ uses
$O(\alpha_j^{-2}\log(J/\delta))$ samples.  The level cost is therefore
\[
\widetilde O\!\left(
 nha_j\cdot\frac{J^2}{\varepsilon^2a_j^2}
\right)
=\widetilde O\!\left(\frac{nh}{\varepsilon^2a_j}\right).
\]
Since $\sum_j1/a_j<2$, the total has the claimed order.  The sum of the
estimation errors after division by $a_j$, plus the tail in~\eqref{eq:dyadic},
is below $\varepsilon/2$ after adjusting constants.
\end{proof}

\begin{proof}[Proof of Theorem~\ref{thm:c-supplied}]
Sample a window and a uniform potential edge.  The corresponding Bernoulli
mean is $\bar e/M$, and $\bar e\le K+n$.  Bernstein's inequality gives
$\widetilde O((K+n)/\varepsilon^2)$ samples for additive edge-count error
$\varepsilon n/4$; the usual additive Bernstein term is absorbed because
$K+n\ge n$ and $\varepsilon\le1$.  Each sample costs at most $m$ snapshot
queries.  Proposition~\ref{prop:c-dyadic} handles components.  Alternatively,
read all $O(n^2m)$ bits.  Take the cheaper branch.
\end{proof}

\subsection{Full proof of quantum component estimation}

\begin{proposition}[Coherent dyadic component estimator]
\label{prop:q-dyadic-full}
Given coherent adjacency-bit access to an $n$-vertex graph $H$, one can
estimate $c(H)$ to additive error $\varepsilon n/2$ with failure at most
$\delta$ using $\widetilde O(\sqrt n/\varepsilon)$ adjacency queries, for
$1/n\le\varepsilon\le1/4$.
\end{proposition}

\begin{proof}
Use the decomposition~\eqref{eq:dyadic}.  For level $j$, explore the component
of a uniform vertex, halting as soon as $L_j=\min\{n,2a_j\}$ vertices have been
discovered: a $2a_j$-th discovery certifies $s\ge2a_j$, whereas an exhausted
queue below the cap determines $s$ exactly.  Process discovered vertices
in queue order.  To process $u$, repeatedly use bounded-error quantum search
over $w\in[n]$ for an edge $uw$ whose second endpoint is not in the explicit
discovered list.  Verify returned candidates.  There are at most $L_j$
successful discoveries and $L_j$ terminating failed searches, so a padded
fixed circuit uses $\widetilde O(a_j\sqrt n)$ adjacency queries.

If the exploration finishes with size $s\in[a_j,2a_j)$, output a reversible
coin with mean $a_j/s$; otherwise output zero.  The ideal mean is $\mu_j$.
Choose
\[
\alpha_j=\frac{\varepsilon a_j}{32J}.
\]
Set the complete bounded-exploration and arithmetic expectation bias to at
most $\alpha_j/2$, and amplitude-estimate the actual circuit's output mean to
additive error $\alpha_j/2$, with failure at most $\delta/(2J)$.  This takes
$O(\alpha_j^{-1}\log(J/\delta))$ calls to the sampler and its inverse.
The query cost at level $j$ is therefore
\[
\widetilde O\!\left(
 a_j\sqrt n\cdot\frac{J}{\varepsilon a_j}
\right)=\widetilde O(\sqrt n/\varepsilon).
\]
There are $J=O(\log(1/\varepsilon))$ levels, absorbed by the tilde.
On the joint success event,
\[
\left|\sum_j\frac{\widehat\mu_j-\mu_j}{a_j}\right|
\le\sum_j\frac{\alpha_j}{a_j}\le\varepsilon/32,
\]
and the tail is at most $\varepsilon/16$.  Adjusting the fixed constants leaves
error below $\varepsilon/2$.  The circuit estimates its actual reversible
sampler, and failed internal searches enter only as expectation bias.
\end{proof}

For a superposition of windows of length at most $h$, replace each ideal graph
edge query by a coherently amplified search for a missing snapshot.  Compile
its clean compute/copy/uncompute implementation to operator error
$O(\delta/Q_{\max})$, where $Q_{\max}$ is the deterministic ideal-query cap.
A hybrid argument then multiplies the query bound by
$\widetilde O(\sqrt h)$ and changes the total failure by only the allocated
amount.

\subsection{Classical density bootstrap without a supplied bound}

\begin{proposition}[Adaptive classical average-density estimator]
\label{prop:c-adaptive}
For a classically sampleable input-independent window law of maximum length
$h$, let $F=\E_I\beta_1(G_I)$.  There is a randomized estimator with
success at least $1-\delta$, deterministic cap $O(n^2h)$, expected query cost
\[
\widetilde O\!\left(
 h\min\left\{n^2,\frac{F+n}{\varepsilon^2}\right\}
\right),
\]
and the same instance-sensitive bound outside an event of probability
$O(\delta/n^3)$.
\end{proposition}

\begin{proof}
As in Theorem~\ref{thm:q-adaptive}, augment the $M=\binom n2$ potential-edge
labels by $n$ known marked dummies.  A uniform window/index sample has marked
probability $E'/(M+n)$, where $E'=\bar e+n$ and
$F+n\le E'\le F+2n$.  Since $E'/(M+n)=\Omega(1/n)$, standard Bernoulli
sampling gives a constant-factor estimate of $E'$ using
$O(n\log(n/\delta))$ samples.  Reduce the bootstrap failure to
$\delta_0=\Theta(\delta/n^3)$ and form the same clipped constant-factor upper
bound $U$.

On the good event, Bernstein edge estimation uses
$\widetilde O(U/\varepsilon^2)$ samples, and
Proposition~\ref{prop:c-dyadic} uses
$\widetilde O(n/\varepsilon^2)$ additional ideal intersection queries.  Every
sample costs at most $h$ snapshot queries.  Cap the procedure by reading all
$O(n^2h)$ relevant bits.  The bad-event contribution to expectation is
$O(\delta_0n^2h)$, absorbed by the good-event bound because $F+n\ge n$ and
$\varepsilon\le1$.  The same good bootstrap event yields the stated
high-probability cost.
\end{proof}

\section{Proofs for Section~\ref{sec:dense}}\label{app:dense}

\subsection{Metric-output adversary lemma}

\begin{proposition}[Biregular adversary for real-valued approximation]
\label{prop:metric-adversary}
Let $X,Y$ be finite sets of bit strings and let
$\mathcal R\subseteq X\times Y$ be biregular, with degrees $d_X,d_Y$.
For each query position $i$, suppose every $x\in X$ has at most $s_X$
related $y$ differing at $i$, and every $y\in Y$ has at most $s_Y$ related
$x$ differing at $i$.  Let $A_x$ denote the set of classical outputs accepted
as correct on input $x$.  If $A_x\cap A_y=\varnothing$ for every
$(x,y)\in\mathcal R$, then every quantum algorithm that outputs an element of
$A_x$ with probability at least $2/3$ on every input uses
\[
\Omega\!\left(\sqrt{d_Xd_Y/(s_Xs_Y)}\right)
\]
queries.  The unions $\bigcup_{x\in X}A_x$ and $\bigcup_{y\in Y}A_y$ may
overlap.
\end{proposition}

\begin{proof}
Let $\Gamma$ be the rectangular adjacency matrix of $\mathcal R$, and let
$u,v$ be the uniform unit vectors on $X,Y$.  Biregularity gives
$u^T\Gamma v=\|\Gamma\|=\sqrt{d_Xd_Y}$: the displayed singular value follows
from the uniform vectors, while the row-sum/column-sum norm bound gives the
reverse inequality.  Purify the algorithm and define after $k$ queries
\[
W_k=\sum_{x,y}\Gamma_{xy}u_xv_y
\langle\psi_x^k\mid\psi_y^k\rangle.
\]
Initially $W_0=\|\Gamma\|$.  Input-independent unitaries do not alter it.

Fix a related pair.  Let $P_x$ be the final measurement event that the output
lies in $A_x$.  Its probability is at least $2/3$ on $x$ and at most $1/3$ on
$y$, because $A_x\cap A_y=\varnothing$ and the algorithm is correct on $y$.
Splitting the inner product across $P_x$ and its complement gives
\[
|\langle\psi_x^Q\mid\psi_y^Q\rangle|
\le\sqrt{(2/3)(1/3)}+\sqrt{(1/3)(2/3)}=2\sqrt2/3.
\]
Therefore $|W_Q|\le(2\sqrt2/3)\|\Gamma\|$.

Let $\Gamma_i$ retain the related pairs that differ at query position $i$.
If $\alpha_{x,i}$ is the norm of the query-index-$i$ component immediately
before a query, a bit query changes the inner product of a pair by at most
$2\sum_{i:x_i\ne y_i}\alpha_{x,i}\alpha_{y,i}$.  Put
$(a_i)_x=u_x\alpha_{x,i}$ and $(b_i)_y=v_y\alpha_{y,i}$.  Then
\[
|W_{k+1}-W_k|
\le2\sum_i a_i^T\Gamma_i b_i
\le2\max_i\|\Gamma_i\|
\sqrt{\sum_i\|a_i\|^2\sum_i\|b_i\|^2}.
\]
Both sums of squared norms equal one.  The row/column degree bounds give
$\|\Gamma_i\|\le\sqrt{s_Xs_Y}$.  A constant fraction of
$\sqrt{d_Xd_Y}$ must therefore be removed at rate at most
$2\sqrt{s_Xs_Y}$ per query, proving the claim.
\end{proof}

\subsection{Proof of Theorem~\ref{thm:dense}}

\begin{proof}[Proof of Theorem~\ref{thm:dense}]
The upper bounds follow from Theorems~\ref{thm:q-supplied} and
\ref{thm:c-supplied} with $K=B_n=\Theta(n^2)$; the classical bound may also
read the complete input.

For the lower bounds, keep a spanning star at every time and associate one
optional leaf-to-leaf edge with each of
$B_n=\binom{n-1}{2}$ temporal rows.  A row is either all one or has one zero at
a time
\[
\mathcal T_m=\{\lceil m/3\rceil,\ldots,\lfloor2m/3\rfloor\},
\qquad L=|\mathcal T_m|\ge m/4.
\]
The star plus one optional edge is a fundamental cycle independent of all
other optional edges.  Hence an all-one row gives one bar of lifetime $m$,
whereas a zero at time $t$ gives two bars of lifetimes $t-1$ and $m-t$.
For $r\ge2$ its $S_r$ contribution is at most $5/9$; for $A$ the lost mass is
at least $4/9$.  We use the common loss lower bound $1/3$.

For the quantum bound, let $a=\lfloor B_n/2\rfloor$,
$d=\lfloor6\varepsilon n\rfloor+1$, and $b=a+d$.  Relate inputs with $a$
all-one rows to inputs with $b$ all-one rows by erasing the zeros in $d$ rows.
Every related output pair differs by more than $2\varepsilon n$.  The relation
degrees are
\[
d_X=\binom{B_n-a}{d},\quad d_Y=\binom bdL^d,
\]
and a fixed query position changes at most
\[
s_X=\binom{B_n-a-1}{d-1},\quad
s_Y=\binom{b-1}{d-1}L^{d-1}
\]
related neighbors.  The positive adversary progress bound for pairwise
disjoint acceptable output intervals gives
\[
\Omega\!\left(\sqrt{\frac{d_Xd_Y}{s_Xs_Y}}\right)
=\Omega\!\left(\frac{\sqrt{(B_n-a)bL}}d\right)
=\Omega(n\sqrt m/\varepsilon).
\]
The event defining the acceptable output interval may depend on the related
pair; a global Boolean threshold separating all $X$ and $Y$ values is not
needed.

For the classical bound, independently make each row all one with probability
$p_\pm=1/2\pm48/n$, and otherwise place its unique zero uniformly in
$\mathcal T_m$.  Let $w(t)$ be the split-row contribution and
$\mu=\E_t[1-w(t)]\ge1/3$.  The two statistic means lie on opposite sides of
$B_n(1-\mu/2)$ at distance at least $4n$, while the total variance is at most
$B_n/4\le n^2/8$.  A pointwise-correct estimator, after constant median
amplification, therefore distinguishes the two distributions with constant
advantage.

After $q$ distinct one-answers in one unresolved row, the probability that a
fresh central query returns zero is
\[
z_p(q)=\frac{1-p}{L-(1-p)q}.
\]
For $p\in[1/4,3/4]$, one has $z_p\ge1/(4L)$,
$1-z_p\ge1/4$, and $|\partial z_p/\partial p|\le16/L$.  The conditional
Bernoulli KL divergence of one informative query is therefore
$O(1/(n^2L))$.  Queries after discovering the zero, repeated queries, and
fixed bits contribute zero.  The transcript chain rule and Pinsker's
inequality imply $\Omega(n^2L)=\Omega(n^2m)$ queries.
\end{proof}

\section{Proofs for Section~\ref{sec:sparse}}\label{app:sparse}

\begin{proof}[Quantum upper bound of Theorem~\ref{thm:low-density}]
Theorem~\ref{thm:q-supplied} gives
$\widetilde O(\sqrt{m(K+n)}/\varepsilon)$ queries.  Since $K\ge256n$, this is
$\widetilde O(\sqrt{mK}/\varepsilon)$.
\end{proof}

\begin{proof}[Quantum lower bound of Theorem~\ref{thm:low-density}]
Put $B=B_n$.  Use the star construction from the proof of
Theorem~\ref{thm:dense}.  Every optional temporal row is either all one or has
one zero at a central time.  Set
\[
d=\lceil4\varepsilon n\rceil,
\qquad b=\lfloor K/4\rfloor,
\qquad a=b-d.
\]
The parameter assumptions imply $a\ge1$, $b\le B/32$, and $d\le b$.
Let $X$ contain inputs with exactly $a$ all-one rows and $Y$ inputs with
exactly $b$ all-one rows.  Relate $x\in X$ to $y\in Y$ when $y$ is obtained by
erasing the unique zero in exactly $d$ split rows.

All related inputs satisfy the promise.  Indeed, an input in $Y$ has statistic
at most
\[
b+\lambda_{\phi,m}(B-b)
\le K/4+\varepsilon n/64<K.
\]
Moreover,~\eqref{eq:leakage-promise} and $B\ge n$ imply
$\lambda_{\phi,m}\le1/64$.  Hence every related pair satisfies
\[
F_\phi(y)-F_\phi(x)
\ge d(1-\lambda_{\phi,m})
>2\varepsilon n.
\]
Their additive-error output intervals are disjoint.

The relation parameters are
\[
d_X=\binom{B-a}{d},\quad d_Y=\binom bdL^d,
\quad s_X=\binom{B-a-1}{d-1},\quad
s_Y=\binom{b-1}{d-1}L^{d-1}.
\]
The metric-output adversary argument used in
Theorem~\ref{thm:dense} therefore gives
\[
\Omega\!\left(
\frac{\sqrt{(B-a)bL}}d
\right)
=\Omega\!\left(\frac{\sqrt{mK}}\varepsilon\right),
\]
because $B-a=\Theta(n^2)$, $b=\Theta(K)$,
$L=\Theta(m)$, and $d=\Theta(\varepsilon n)$.
\end{proof}

\begin{proposition}[Adaptive censored-row information bound]
\label{prop:censored-kl}
Fix $L\ge4$.  A temporal row is all one with probability $p$ and otherwise
contains one zero at a uniformly random position in $[L]$.  Rows are
independent.  Let $0<p_-<p_+\le1/2$, put $\Delta=p_+-p_-$, and assume
$\Delta\le p_-/2$.  For every possibly randomized adaptive algorithm making
at most $Q$ bit queries, if $\mathsf T_+$ and $\mathsf T_-$ are its complete
transcripts under the two product distributions, then
\[
D_{\mathrm{KL}}(\mathsf T_+\|\mathsf T_-)
\le128Q\frac{\Delta^2}{p_-L}.
\]
\end{proposition}

\begin{proof}[Proof of Proposition~\ref{prop:censored-kl}]
Include the algorithm's random coins in the transcript.  Conditional on them,
the query policy is deterministic, so it suffices to prove the claim for a
deterministic policy and then average.

Augment the recorded transcript, without revealing any extra information to
the algorithm, as follows.  When the original query process has accumulated
$h=\lfloor L/2\rfloor$ distinct one-answers in a row, append the latent row
type and, for a nonfull row, its zero location to the analysis transcript.  The
algorithm continues to follow its original policy and does not see this extra
record.  The original transcript is the projection of the augmented one, so
data processing permits us to upper-bound its KL divergence by that of the
augmented transcript.

Before a reveal, suppose $q<h$ distinct queried positions in a row have all
returned one.  A fresh query is zero with conditional probability
\[
z_p(q)=\frac{1-p}{L-(1-p)q}.
\]
Since $q<L/2$ and $p\le1/2$,
\[
z_p(q)\ge\frac1{2L},\qquad z_p(q)\le\frac2L\le\frac12,
\qquad
\left|\frac{\partial z_p(q)}{\partial p}\right|
=\frac{L}{[L-(1-p)q]^2}\le\frac4L.
\]
Using
$D_{\mathrm{KL}}(\operatorname{Ber}(u)\|\operatorname{Ber}(v))
\le(u-v)^2/[v(1-v)]$, each ordinary informative answer contributes at most
$64\Delta^2/L$, which is at most
$64\Delta^2/(p_-L)$.

Now consider a reveal.  Put $c=(L-h)/L\ge1/2$.  Conditional on $h$ one-answers,
the posterior probability of an all-one row is
\[
a_p=\frac{p}{p+(1-p)c}
=\frac{p}{c+(1-c)p}.
\]
The zero location conditional on a nonfull row is uniform among the unqueried
positions under both hypotheses, so the reveal KL is exactly the KL between
$\operatorname{Ber}(a_{p_+})$ and $\operatorname{Ber}(a_{p_-})$.
Since $|a'_p|\le2$, $a_{p_-}\ge p_-$, and
$1-a_{p_-}\ge1/4$, this is at most $16\Delta^2/p_-$.  Every reveal consumes at
least $h\ge L/4$ previous queries, so at most $4Q/L$ reveals occur.  Their
total KL is at most $64Q\Delta^2/(p_-L)$.

After a zero or a reveal, the row is resolved and subsequent queries add no
information.  The transcript chain rule adds the ordinary-query and reveal
contributions, giving the stated constant 128.
\end{proof}

\begin{proof}[Randomized classical lower bound for
Theorem~\ref{thm:low-density}]
Let $B=B_n$, $\lambda=\lambda_{\phi,m}$, and let $w(t)$ be the split-row
contribution.  Define product input distributions $\mathcal D_-$ and
$\mathcal D_+$: independently for each of the $B$ optional rows, make it all
one with probability
\[
p_- =\frac{K/4}{B},
\qquad
p_+=\frac{K/4+D}{B},
\qquad
D=16(\varepsilon n+\sqrt K),
\]
and otherwise put one zero at a uniform central time.  The numerical
assumptions imply
\[
D\le K/64,
\qquad p_+\le1/2,
\qquad p_+-p_-\le p_-/2.
\]
Let $Y_j\in[0,1]$ be the contribution of row $j$, let
$F=\sum_jY_j$, and put $\bar w=\E_tw(t)\le\lambda$.  The means are
\[
\mu_\pm=B[\bar w+(1-\bar w)p_\pm],
\]
so
\[
\mu_+-\mu_-=(1-\bar w)D\ge(63/64)D.
\]
Moreover,
\[
\operatorname{Var}(F)
\le B(p_++\lambda^2)<K/3.
\]
The last inequality uses $Bp_+=K/4+D\le17K/64$ and
$B\lambda^2\le(\varepsilon n/64)^2/B\le1/65536$.

Let
\[
\mathcal G_\pm=\{|F-\mu_\pm|\le4\sqrt K\}.
\]
Chebyshev gives $\Prb_{\mathcal D_\pm}(\mathcal G_\pm^c)<1/32$.
Every input in either good event satisfies the promise: using
$4\sqrt K\le K/512$,
\[
F\le\mu_++4\sqrt K
\le17K/64+\varepsilon n/64+K/512<K.
\]
Let $\tau=(\mu_++\mu_-)/2$.  On the two good events, the distance from $F$ to
$\tau$ is at least
\[
\frac{\mu_+-\mu_-}{2}-4\sqrt K
>\varepsilon n.
\]
Thus thresholding any additive-$\varepsilon n$ estimate distinguishes the two
conditioned, promise-supported distributions with error at most $1/3$.
Their output-transcript total variation is at least $1/3$.

Let $\mathsf T_\pm$ be the unconditioned transcripts.  Conditioning changes
each input distribution, and hence each transcript distribution, by total
variation at most $1/32$.  Therefore
\[
d_{\mathrm{TV}}(\mathsf T_+,\mathsf T_-)
\ge\frac13-\frac1{16}>\frac14.
\]
Pinsker gives $D_{\mathrm{KL}}(\mathsf T_+\|\mathsf T_-)\ge1/8$.
Proposition~\ref{prop:censored-kl}, with
$\Delta=D/B$, now implies
\[
Q=\Omega\!\left(\frac{p_-L}{\Delta^2}\right)
=\Omega\!\left(\frac{KBL}{D^2}\right)
=\Omega\!\left(
\frac{mKn^2}{K+\varepsilon^2n^2}
\right).
\]
For positive $x,y$, $xy/(x+y)\ge\frac12\min\{x,y\}$; hence
\[
Q=\Omega\!\left(
 m\min\left\{n^2,\frac K{\varepsilon^2}\right\}
\right).
\]
This gives a direct randomized information lower bound.  Equivalently, the equal
mixture of the two conditioned distributions is a valid Yao hard distribution:
pointwise correctness is used only on their promise-supported good events, while
the information bound is applied to the original independent product
distributions.
\end{proof}

\begin{proof}[Randomized classical upper bound for
Theorem~\ref{thm:low-density}]
Theorem~\ref{thm:c-supplied} and $K\ge256n$ give
\[
\widetilde O\!\left(
 m\min\{n^2,K/\varepsilon^2\}
\right).
\]
Together with the lower bound, this completes the proof.
\end{proof}

\section{Proofs for Sections~\ref{sec:weights} and~\ref{sec:histogram}}\label{app:weights}

\begin{corollary}[Logarithmic power weights]
For $\phi_r(x)=x^r$,
\[
\lambda_{\phi_r,m}\le2(2/3)^r.
\]
Consequently the leakage promise~\eqref{eq:leakage-promise} holds whenever
\[
r\ge
\left\lceil\log_{3/2}\frac{128B_n}{\varepsilon n}\right\rceil.
\]
The simpler sufficient choice
$r\ge\lceil\log_{3/2}(64n/\varepsilon)\rceil$ follows from
$B_n\le n^2/2$.
\end{corollary}
\begin{proof}[Proof of Corollary~\ref{cor:power}]
For a central split, both normalized piece lengths are at most $2/3$.
Monotonicity of $x^r$ gives the leakage bound; rearrange it to obtain the
stated degree.
\end{proof}

\begin{corollary}[Binomial soft-long-bar weight]
For $\phi_s^{\mathrm{bin}}$,
\[
\lambda_{\phi_s^{\mathrm{bin}},m}
\le\frac43(5/6)^s.
\]
Thus~\eqref{eq:leakage-promise} holds whenever
\[
s\ge
\left\lceil\log_{6/5}\frac{256B_n}{3\varepsilon n}\right\rceil,
\]
and it is sufficient to take
$s\ge\lceil\log_{6/5}(128n/(3\varepsilon))\rceil$.
\end{corollary}
\begin{proof}[Proof of Corollary~\ref{cor:binomial}]
For $x\le2/3$,
$x((1+x)/2)^s\le(2/3)(5/6)^s$.  Sum the two split pieces and rearrange.
\end{proof}

\subsection{Resource accounting}\label{sec:resources}

For a fixed power $r$, prepare $r$ independent uniform time registers and
compute their minimum and maximum by a reversible comparator tree.  For the
binomial weight, prepare $s$ Hadamard control bits, one mandatory time register,
and $s$ conditionally active uniform time registers.  The active count is
$1+\operatorname{Bin}(s,1/2)$ exactly.  With $m$ not a power of two, an
input-independent reversible rejection or amplitude-amplification circuit
prepares a uniform label in $[m]$ to any desired accuracy using polylogarithmic
gates; no graph data or QRAM is involved.

A window-law total-variation error $\zeta$ changes the target expectation by at
most $B_n\zeta$.  It therefore suffices to compile the sampling distribution
to $\zeta=O(\varepsilon n/B_n)=O(\varepsilon/n)$.  Inside amplitude estimation,
each coherent subroutine is compiled to operator error inverse-polynomial in
the deterministic total-use cap so that a hybrid bound fits the allocated
failure budget.

For both explicit weights, one window preparation uses
\[
C_\phi=\widetilde O(r_{\max}\log m)
\]
gates and $O(r_{\max}\log m)$ qubits; a parallel comparator tree has depth
$\widetilde O(\log r_{\max}\log m)$ at the cost of those registers.  For the
supplied-$K$ quantum algorithm, a conservative gate-work bound is
\[
\widetilde O\!\left[
\frac{\sqrt{m(K+n)}}\varepsilon(C_O+\log m)
+\frac{\sqrt{K+n}}\varepsilon C_\phi
+\frac{\sqrt n}{\varepsilon^2}
\right].
\]
The last term is the explicit-list membership overhead in the dyadic component
sampler.  A straightforward implementation uses live data space
\[
O\!\left(r_{\max}\log m+\varepsilon^{-1}\log n
+\polylog(nm/(\varepsilon\delta))\right)
\]
qubits, in addition to oracle workspace.  The sequential depth is at most the
gate-work bound; no parallel-depth separation is claimed.  Since
$0\le F_\phi\le B_n$, the classical scalar output needs
$O(\log(n/\varepsilon))$ bits.

\begin{proof}[Proof of Theorem~\ref{thm:statistic-nonrecovery}]
Let $A\in\mathbb Q^{k\times m}$ have
$A_{j\ell}=\phi_j(\ell/m)$.  Since $k<m$, choose a nonzero rational vector
$z\in\ker A$ and scale it to a nonzero integer vector.  Every entry of every
column of $A$ is positive because the weights have no constant term and
$\ell/m>0$.  Hence $z$ has both positive and negative entries.  Write
$z=z^+-z^-$ with distinct nonnegative integer vectors $z^+,z^-$.  Then
$Az^+=Az^-$, so these two positive-lifetime histograms give identical
weighted-statistic vectors.

Realize either histogram by keeping a spanning star at every time and assigning
one distinct optional leaf-to-leaf edge to each requested bar.  To realize a
bar of lifetime $\ell$, make its optional edge present on one contiguous block
of $\ell$ snapshots and absent elsewhere.  Fundamental cycles for distinct
optional edges are independent, and the intersection junctions make each such
cycle an interval summand on exactly that block. Choose enough vertices to
provide all required optional edges. The two resulting positive-lifetime
histograms are different but have equal weighted statistics.
\end{proof}

Exact power moments $S_1,\ldots,S_m$ do determine the positive-lifetime
histogram by Vandermonde inversion. This does not yield a useful free recovery theorem:
it requires $m$ scalar outputs, exact or sufficiently high-precision values,
and still cannot recover bar positions from lifetimes alone.

\section{Finite Sanity Checks}\label{app:checks}

As a guard against indexing, normalization, and promise-support mistakes, we
carried out the following exact-rational or conservative numerical checks:
\begin{enumerate}
\item For $m\le12$, direct enumeration of all $r$-tuples of time samples agrees
with $(\ell/m)^r$, and the uniform-window weight agrees exactly with
$(x+mx^2)/(m+1)$.
\item For $3\le m<80$, all central split weights for $2\le r\le8$ and the
uniform-window weight are at most $5/9$; the integer-rounding endpoints in the
logarithmic-degree bounds for power and binomial weights were checked
separately.
\item At boundary values $n=2^{14}$,
$\varepsilon\in\{1/n,1/128,1/8\}$, and
$K\in\{256n,512n,\lfloor B_n/8\rfloor\}$, the inequalities
$D\le K/64$, $p_+\le1/2$, $p_+-p_-\le p_-/2$, promise support, mean separation,
and the Chebyshev budget in Theorem~\ref{thm:low-density} all hold.
\item A dense grid of $L,p_-,p_+,q$ values satisfies the numerical KL upper
bounds $64\Delta^2/L$ before reveal and $16\Delta^2/p_-$ at reveal in
Proposition~\ref{prop:censored-kl}.
\item The dyadic reciprocal-size expansion has nonnegative tail at most
$2^{-J}$ for all $2\le n<200$ and the precision values used in the component
estimator.
\end{enumerate}
These checks are not substitutes for the proofs; they only provide finite
sanity tests for the parameter ranges stated above.

\section*{AI Disclosure}
We used ChatGPT and Claude to assist with literature review, brainstorming,
proof checking, and language polishing (Sections~1, 4, 9, and~10, and
Appendices~A--E).  The authors independently determined the research direction
and final presentation, verified all mathematical arguments and references,
and take full responsibility for the correctness, originality, and final
content.

\end{document}